\documentclass[11pt,a4paper]{article}
\pdfoutput=1
\usepackage[latin1]{inputenc}
\usepackage{hyperref}
\usepackage{amsmath,amsthm,cite,color,eucal,graphicx}
\usepackage[expansion=false]{microtype}

\newcommand{\myetal}{et al.\mbox{}}
\newcommand{\mysize}[1]{{\lvert #1 \rvert}}

\newcommand{\myG}{\mathcal{G}}

\newtheorem{theorem}{Theorem}
\newtheorem{lemma}{Lemma}

\theoremstyle{definition}
\newtheorem{definition}{Definition}

\theoremstyle{remark}

\hypersetup{
    colorlinks=false,
    pdfborder={0 0 0},
    pdftitle={Almost stable matchings in constant time},
    pdfauthor={Patrik Floréen, Petteri Kaski, Valentin Polishchuk, Jukka Suomela},
}

\begin{document}

\title{Almost Stable Matchings in Constant Time}
\newcommand{\bat}{@\allowbreak}
\newcommand{\bdot}{.\allowbreak}
\author{Patrik Flor\'{e}en, Petteri Kaski, \\ Valentin Polishchuk, and Jukka Suomela\thanks{%
    Helsinki Institute for Information Technology HIIT,
    Helsinki University of Technology and University of Helsinki.
    \emph{Address:}
    P.O. Box 68,
    FI-00014 University of Helsinki,
    Finland.
    \emph{Email addresses:}
    patrik\bdot floreen\bat cs.helsinki.fi,
    petteri\bdot kaski\bat cs.helsinki.fi,
    valentin\bdot polishchuk\bat cs.helsinki.fi,
    jukka\bdot suomela\bat cs.helsinki.fi.
}}
\date{December 2008}
\maketitle
\begin{abstract}

We show that the ratio of matched individuals to blocking pairs grows linearly with the number of propose--accept rounds executed by the Gale--Shapley algorithm for the stable marriage problem. Consequently, the participants can arrive at an almost stable matching even without full information about the problem instance; for each participant, knowing only its local neighbourhood is enough. In distributed-systems parlance, this means that if each person has only a constant number of acceptable partners, an almost stable matching emerges after a constant number of synchronous communication rounds. This holds even if ties are present in the preference lists.

We apply our results to give a distributed $(2+\epsilon)$-approximation algorithm for maximum-weight matching in bicoloured graphs and a centralised randomised constant-time approximation scheme for estimating the size of a stable matching.

\end{abstract}

\section{Introduction}

The social networking boom brings up computational challenges unforeseen in the past. In a modern large-scale network, gathering full information about the whole network in one place is practically impossible. This motivates design and analysis of \emph{local algorithms} \cite{linial92locality,naor95what,suomela08survey} in which the output of a node depends only on the input within a constant number of edges (hops) from the node.

In this paper we study a local version of the classical Gale--Shapley algorithm \cite{gale62college} for the \emph{stable marriage problem}. We show that early termination of the algorithm leads to a matching with relatively few unstable edges.

\subsection{Matchings}\label{sec:definition}

We follow the convention (see, e.g., Fleiner \cite{fleiner03fixed-point}) that an instance of the stable marriage problem is specified by a simple undirected bipartite graph $\myG=(R \cup B, E)$; the graph may be disconnected but there are no isolated nodes. Call the nodes in $R$ {\em red}, the nodes in $B$ {\em blue}, and the graph $\myG$ a {\em bicoloured} graph. Each node $v\in R \cup B$ has a linear order on the adjacent nodes. The linear order constitutes the {\em matching preference} for $v$; if a node $u$ is not adjacent to $v$, then neither $u$ is an acceptable partner for $v$ nor $v$ is acceptable for $u$.

A set of edges $M\subseteq E$ is a {\em matching} if every node is incident to at most one edge in $M$. Two nodes joined by an edge in $M$ are {\em matched} in $M$; a node not incident to an edge in $M$ is {\em unmatched} in $M$. A matching is {\em maximal} if it is not a subset of a larger matching. A matching is {\em maximum} if it has the maximum size among all the matchings. In an edge-weighted graph, a {\em maximum-weight} matching is the one that has the maximum weight among all the matchings. A {\em greedy} matching is a maximal matching obtained by adding the edges one by one, in order of decreasing weight. It is well-known that the greedy matching is a $2$-approximation to the maximum-weight one.

Let $M$ be a matching. An edge $\{u,v\} \in E\setminus M$ is {\em unstable} relative to $M$ if both (i)~$u$ is unmatched or prefers $v$ over its current match in $M$; and (ii)~$v$ is unmatched or prefers $u$ over its current match in $M$. The matching $M$ is {\em stable} if there are no unstable edges. Let $\epsilon > 0$ be some constant. This work is centred on the following notion of ``almost stability''.
\begin{definition}
A matching $M$ is \emph{$\epsilon$-stable} if the number of unstable edges is at most $\epsilon \mysize{M}$.
\end{definition}

Denote by $\Delta$ the maximum degree of a node of $\myG$. That is, $\Delta$ is the maximum number of acceptable partners for one participant. In a social network, the bound on the degree of a node $v$ may not necessarily be due to $v$ being particularly picky; the bound may just reflect the fact that the number of participants, information about whom is available (or comprehensible, or relevant) to $v$, is limited. Throughout this work we assume that $\Delta$ is a constant.

\subsection{Model of Distributed Computing}\label{sec:model}

We view $\myG$ as the communication graph of a distributed system: if $\{u,v\} \in E$, then the node $u$ and the node $v$ can exchange messages.

Let $T$ be a constant that determines the trade-off between the stability and the running time. Initially, each node $v \in R \cup B$ knows the following information: its colour (i.e., whether $v \in R$ or $v \in B$), its degree $d(v)$, and the constant $T$. The node $v$ has $d(v)$ ports through which it can communicate with its neighbours; the ports are numbered according to $v$'s preferences.

We assume synchronous communication. On each time step, every node can (i)~receive messages from its neighbours, (ii)~perform deterministic local computations, and (iii)~send a message to each of its neighbours. The algorithm runs for $T$ synchronous time steps, after which each node needs to produce the output: which of the neighbours, if any, is its partner in the matching. We also require that the output is consistent: if $u$ announces that $v$ is its partner, then $v$ must announce that $u$ is its partner.

\subsection{Statement of Results}\label{sec:results}

Our main contribution lies in an analysis of the ``transient phase'' of the Gale--Shapley algorithm for the stable marriage problem. We truncate the algorithm to a specified number of proposal rounds and investigate the resulting matching for stability. We find that a constant number of rounds suffice to obtain a constant ratio of unstable edges to matching edges. We obtain the following result.
\begin{theorem}\label{thm:almost-stable}
There exists a deterministic distributed algorithm that finds an $\epsilon$-stable matching in a bicoloured graph in time $T \leq 4 + 2\Delta^2/\epsilon$.
\end{theorem}
The algorithm works even in the case when ties in the preference lists are allowed.

We also consider the case where $\myG$ is an edge-weighted graph, and the preferences are determined by the edge weights: each node prefers its incident edges in order of their weight, with the heaviest edge being the most preferred. If all weights are distinct, the unique {\em stable} matching in $\myG$ is the greedy one, and its weight is at least $1/2$ of that of the maximum-weight matching. In general, an {\em almost} stable matching can be a poor approximation to the maximum-weight one; say, a super-heavy edge may be missing from the almost stable matching. Fortunately, the particular almost stable matching from Theorem~\ref{thm:almost-stable} is much better in this sense.
\begin{theorem}\label{thm:weighted-matching}
There exists a deterministic distributed algorithm that finds a ${(2+\epsilon)}$-approximation for maximum-weight matching in bicoloured graphs in time $T \leq 4 + 2\Delta/\epsilon$.
\end{theorem}

Given access to a bicoloured graph via a {\em preference oracle} (queried with a node, the oracle replies with the colour of the node and the adjacent nodes listed in order of preference), the algorithm in Theorem~\ref{thm:almost-stable} enables estimation of the size of a stable matching using a constant number of queries to the oracle. We analyse the non-degenerate case $\Delta \ge 3$.
\begin{theorem}\label{thm:size-estimate}
For any\/ $0<\delta\leq 1/2$, $0 < \epsilon \le 1$, and $\Delta \ge 3$, there exists a randomised algorithm that, given access to a preference oracle for a bicoloured graph\/ $\myG$ without isolated nodes, makes at most $25000 \epsilon^{-2} (\Delta-1)^{3 + 4\Delta/\epsilon}\ln\delta^{-1}$ queries to the oracle and outputs with probability at least\/ $1-\delta$ an estimate $\hat{m}$ such that\/ $\bigl|\hat m-\mysize{M}\bigr| \le \epsilon\mysize{M}$ where $M$ is a stable matching in\/~$\myG$.
\end{theorem}

Here it should be stressed that the algorithm in Theorem~\ref{thm:size-estimate} estimates the size of a stable matching, not just an $\epsilon$-stable matching.

\subsection{Overview}

The rest of the paper is structured as follows. In Section~\ref{sec:motivation} we discuss the specific choices that we make in this work, in particular, why and how an almost stable matching differs from a stable matching in a distributed setting. Section~\ref{sec:related} reviews related work. Sections \ref{sec:algorithm} and \ref{sec:analysis} present and analyse the distributed, truncated version of the Gale--Shapley algorithm. Section~\ref{sec:applications} applies this algorithm to prove Theorems \ref{thm:almost-stable}--\ref{thm:size-estimate}.

\section{Motivation and Implications}\label{sec:motivation}

\subsection{Local Algorithms}\label{sec:local}

The running time $T$ of the algorithms presented in this work only depends on the degree bound $\Delta$ and the desired stability or approximation guarantee~$\epsilon$. For a constant $\Delta$ and $\epsilon$, these are constant-time algorithms; the running time is independent of the number of nodes in the network, the diameter of the network, or other global properties. Such constant-time distributed algorithms are known as local algorithms \cite{linial92locality,naor95what,suomela08survey}.

Besides the practical advantage of a low time complexity, our positive results also provide insight into the structure of the stable marriage problem. We mention three interpretations of Theorem~\ref{thm:almost-stable}, from different perspectives; these all follow from the observation that in $T$ synchronous communication steps, information can only be propagated from distance $T$ in the network.
\begin{itemize}
    \newcommand{\myitem}[1]{\item \emph{#1.}}
    \myitem{Robustness and dynamic graphs} The $\epsilon$-stable matching $M$ from Theorem~\ref{thm:almost-stable} is robust in a dynamic network. A change in the network only affects $M$ in the radius-$T$ neighbourhood around the point of change.
    \myitem{Available information} Information within the radius-$T$ neighbourhood of an edge $e$ is sufficient to determine whether $e$ is part of the globally consistent $\epsilon$-stable matching $M$.
    \myitem{Parallelism and circuit complexity} For any graph, we can construct a bounded-fan-in Boolean circuit that maps the matching preferences to an $\epsilon$-stable matching. The depth of the circuit only depends on the constants $\Delta$ and $\epsilon$, not on the size of the graph.
\end{itemize}
As we will discuss in the next section, these properties highlight a substantial difference between stable and almost stable matchings.

\subsection{Almost Stable vs.\ Stable}\label{sec:almost}

While an almost stable matching is robust to change in the preferences, this is not the case for a stable matching. For example, consider the following graph where the numbered edge ends indicate preference rankings (the most preferred match has rank 1).
\begin{center}
    \begin{picture}(0,0)%
\includegraphics{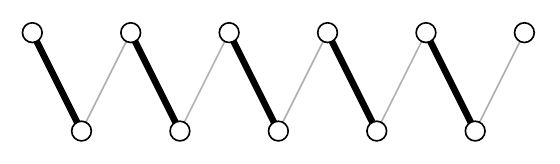}%
\end{picture}%
\setlength{\unitlength}{4144sp}%
\begingroup\makeatletter\ifx\SetFigFontNFSS\undefined%
\gdef\SetFigFontNFSS#1#2#3#4#5{%
  \reset@font\fontsize{#1}{#2pt}%
  \fontfamily{#3}\fontseries{#4}\fontshape{#5}%
  \selectfont}%
\fi\endgroup%
\begin{picture}(2544,744)(-146,242)
\put(361,749){\makebox(0,0)[rb]{\smash{{\SetFigFontNFSS{8}{9.6}{\familydefault}{\mddefault}{\updefault}{\color[rgb]{0,0,0}$1$}%
}}}}
\put( 91,749){\makebox(0,0)[lb]{\smash{{\SetFigFontNFSS{8}{9.6}{\familydefault}{\mddefault}{\updefault}{\color[rgb]{0,0,0}$1$}%
}}}}
\put(541,749){\makebox(0,0)[lb]{\smash{{\SetFigFontNFSS{8}{9.6}{\familydefault}{\mddefault}{\updefault}{\color[rgb]{0,0,0}$2$}%
}}}}
\put(811,749){\makebox(0,0)[rb]{\smash{{\SetFigFontNFSS{8}{9.6}{\familydefault}{\mddefault}{\updefault}{\color[rgb]{0,0,0}$1$}%
}}}}
\put(991,749){\makebox(0,0)[lb]{\smash{{\SetFigFontNFSS{8}{9.6}{\familydefault}{\mddefault}{\updefault}{\color[rgb]{0,0,0}$2$}%
}}}}
\put(1261,749){\makebox(0,0)[rb]{\smash{{\SetFigFontNFSS{8}{9.6}{\familydefault}{\mddefault}{\updefault}{\color[rgb]{0,0,0}$1$}%
}}}}
\put(1711,749){\makebox(0,0)[rb]{\smash{{\SetFigFontNFSS{8}{9.6}{\familydefault}{\mddefault}{\updefault}{\color[rgb]{0,0,0}$1$}%
}}}}
\put(1441,749){\makebox(0,0)[lb]{\smash{{\SetFigFontNFSS{8}{9.6}{\familydefault}{\mddefault}{\updefault}{\color[rgb]{0,0,0}$2$}%
}}}}
\put(1891,749){\makebox(0,0)[lb]{\smash{{\SetFigFontNFSS{8}{9.6}{\familydefault}{\mddefault}{\updefault}{\color[rgb]{0,0,0}$2$}%
}}}}
\put(2161,749){\makebox(0,0)[rb]{\smash{{\SetFigFontNFSS{8}{9.6}{\familydefault}{\mddefault}{\updefault}{\color[rgb]{0,0,0}$1$}%
}}}}
\put(2116,389){\makebox(0,0)[lb]{\smash{{\SetFigFontNFSS{8}{9.6}{\familydefault}{\mddefault}{\updefault}{\color[rgb]{0,0,0}$2$}%
}}}}
\put(1936,389){\makebox(0,0)[rb]{\smash{{\SetFigFontNFSS{8}{9.6}{\familydefault}{\mddefault}{\updefault}{\color[rgb]{0,0,0}$1$}%
}}}}
\put(1666,389){\makebox(0,0)[lb]{\smash{{\SetFigFontNFSS{8}{9.6}{\familydefault}{\mddefault}{\updefault}{\color[rgb]{0,0,0}$2$}%
}}}}
\put(1216,389){\makebox(0,0)[lb]{\smash{{\SetFigFontNFSS{8}{9.6}{\familydefault}{\mddefault}{\updefault}{\color[rgb]{0,0,0}$2$}%
}}}}
\put(1486,389){\makebox(0,0)[rb]{\smash{{\SetFigFontNFSS{8}{9.6}{\familydefault}{\mddefault}{\updefault}{\color[rgb]{0,0,0}$1$}%
}}}}
\put(766,389){\makebox(0,0)[lb]{\smash{{\SetFigFontNFSS{8}{9.6}{\familydefault}{\mddefault}{\updefault}{\color[rgb]{0,0,0}$2$}%
}}}}
\put(766,389){\makebox(0,0)[lb]{\smash{{\SetFigFontNFSS{8}{9.6}{\familydefault}{\mddefault}{\updefault}{\color[rgb]{0,0,0}$2$}%
}}}}
\put(1036,389){\makebox(0,0)[rb]{\smash{{\SetFigFontNFSS{8}{9.6}{\familydefault}{\mddefault}{\updefault}{\color[rgb]{0,0,0}$1$}%
}}}}
\put(136,389){\makebox(0,0)[rb]{\smash{{\SetFigFontNFSS{8}{9.6}{\familydefault}{\mddefault}{\updefault}{\color[rgb]{0,0,0}\textbf{1}}%
}}}}
\put(586,389){\makebox(0,0)[rb]{\smash{{\SetFigFontNFSS{8}{9.6}{\familydefault}{\mddefault}{\updefault}{\color[rgb]{0,0,0}$1$}%
}}}}
\put(316,389){\makebox(0,0)[lb]{\smash{{\SetFigFontNFSS{8}{9.6}{\familydefault}{\mddefault}{\updefault}{\color[rgb]{0,0,0}\textbf{2}}%
}}}}
\end{picture}%

\end{center}
Now transpose the preferences shown in boldface to obtain a graph whose unique stable matching is edge-disjoint from the unique stable matching in the original graph.
\begin{center}
    \begin{picture}(0,0)%
\includegraphics{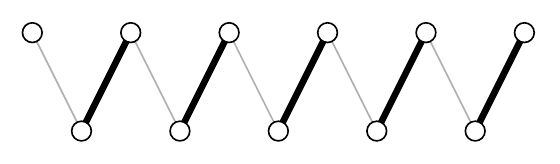}%
\end{picture}%
\setlength{\unitlength}{4144sp}%
\begingroup\makeatletter\ifx\SetFigFontNFSS\undefined%
\gdef\SetFigFontNFSS#1#2#3#4#5{%
  \reset@font\fontsize{#1}{#2pt}%
  \fontfamily{#3}\fontseries{#4}\fontshape{#5}%
  \selectfont}%
\fi\endgroup%
\begin{picture}(2544,744)(-146,242)
\put(586,389){\makebox(0,0)[rb]{\smash{{\SetFigFontNFSS{8}{9.6}{\familydefault}{\mddefault}{\updefault}{\color[rgb]{0.690,0.690,0.690}$1$}%
}}}}
\put(766,389){\makebox(0,0)[lb]{\smash{{\SetFigFontNFSS{8}{9.6}{\familydefault}{\mddefault}{\updefault}{\color[rgb]{0.690,0.690,0.690}$2$}%
}}}}
\put(766,389){\makebox(0,0)[lb]{\smash{{\SetFigFontNFSS{8}{9.6}{\familydefault}{\mddefault}{\updefault}{\color[rgb]{0.690,0.690,0.690}$2$}%
}}}}
\put(1036,389){\makebox(0,0)[rb]{\smash{{\SetFigFontNFSS{8}{9.6}{\familydefault}{\mddefault}{\updefault}{\color[rgb]{0.690,0.690,0.690}$1$}%
}}}}
\put(1216,389){\makebox(0,0)[lb]{\smash{{\SetFigFontNFSS{8}{9.6}{\familydefault}{\mddefault}{\updefault}{\color[rgb]{0.690,0.690,0.690}$2$}%
}}}}
\put(1486,389){\makebox(0,0)[rb]{\smash{{\SetFigFontNFSS{8}{9.6}{\familydefault}{\mddefault}{\updefault}{\color[rgb]{0.690,0.690,0.690}$1$}%
}}}}
\put(1936,389){\makebox(0,0)[rb]{\smash{{\SetFigFontNFSS{8}{9.6}{\familydefault}{\mddefault}{\updefault}{\color[rgb]{0.690,0.690,0.690}$1$}%
}}}}
\put(1666,389){\makebox(0,0)[lb]{\smash{{\SetFigFontNFSS{8}{9.6}{\familydefault}{\mddefault}{\updefault}{\color[rgb]{0.690,0.690,0.690}$2$}%
}}}}
\put(2116,389){\makebox(0,0)[lb]{\smash{{\SetFigFontNFSS{8}{9.6}{\familydefault}{\mddefault}{\updefault}{\color[rgb]{0.690,0.690,0.690}$2$}%
}}}}
\put(1891,749){\makebox(0,0)[lb]{\smash{{\SetFigFontNFSS{8}{9.6}{\familydefault}{\mddefault}{\updefault}{\color[rgb]{0.690,0.690,0.690}$2$}%
}}}}
\put(2161,749){\makebox(0,0)[rb]{\smash{{\SetFigFontNFSS{8}{9.6}{\familydefault}{\mddefault}{\updefault}{\color[rgb]{0.690,0.690,0.690}$1$}%
}}}}
\put(1711,749){\makebox(0,0)[rb]{\smash{{\SetFigFontNFSS{8}{9.6}{\familydefault}{\mddefault}{\updefault}{\color[rgb]{0.690,0.690,0.690}$1$}%
}}}}
\put(1441,749){\makebox(0,0)[lb]{\smash{{\SetFigFontNFSS{8}{9.6}{\familydefault}{\mddefault}{\updefault}{\color[rgb]{0.690,0.690,0.690}$2$}%
}}}}
\put(991,749){\makebox(0,0)[lb]{\smash{{\SetFigFontNFSS{8}{9.6}{\familydefault}{\mddefault}{\updefault}{\color[rgb]{0.690,0.690,0.690}$2$}%
}}}}
\put(1261,749){\makebox(0,0)[rb]{\smash{{\SetFigFontNFSS{8}{9.6}{\familydefault}{\mddefault}{\updefault}{\color[rgb]{0.690,0.690,0.690}$1$}%
}}}}
\put(541,749){\makebox(0,0)[lb]{\smash{{\SetFigFontNFSS{8}{9.6}{\familydefault}{\mddefault}{\updefault}{\color[rgb]{0.690,0.690,0.690}$2$}%
}}}}
\put(811,749){\makebox(0,0)[rb]{\smash{{\SetFigFontNFSS{8}{9.6}{\familydefault}{\mddefault}{\updefault}{\color[rgb]{0.690,0.690,0.690}$1$}%
}}}}
\put(361,749){\makebox(0,0)[rb]{\smash{{\SetFigFontNFSS{8}{9.6}{\familydefault}{\mddefault}{\updefault}{\color[rgb]{0.690,0.690,0.690}$1$}%
}}}}
\put( 91,749){\makebox(0,0)[lb]{\smash{{\SetFigFontNFSS{8}{9.6}{\familydefault}{\mddefault}{\updefault}{\color[rgb]{0.690,0.690,0.690}$1$}%
}}}}
\put(136,389){\makebox(0,0)[rb]{\smash{{\SetFigFontNFSS{8}{9.6}{\familydefault}{\mddefault}{\updefault}{\color[rgb]{0,0,0}\textbf{2}}%
}}}}
\put(316,389){\makebox(0,0)[lb]{\smash{{\SetFigFontNFSS{8}{9.6}{\familydefault}{\mddefault}{\updefault}{\color[rgb]{0,0,0}\textbf{1}}%
}}}}
\end{picture}%

\end{center}
Thus, a single transposition in preferences can force every match to break~up.

From the perspective of distributed computation, this example shows that every node must have essentially complete information about the network to arrive at a stable matching. In particular, if each node initially knows only its neighbours and its own preferences, the number of communication steps required between neighbours is linear in the diameter of the network.

Gathering global information is hardly practical in a very large-scale distributed system, especially if the network undergoes frequent changes. The ease with which almost stable matchings can be computed in a distributed setting comes as a positive surprise.

\subsection{Definition of Almost Stable}\label{sec:def}

Our results are somewhat oblivious to the exact definition of an ``almost stable'' matching. The number of unstable edges (blocking pairs) appears to be generally accepted as a basic measure of instability \cite{abraham06almost-stable,biro08size,eriksson08instability,khuller94on-line}, but one could equally well consider other measures. For example, one can consider the number of nodes that are endpoints of unstable edges.

Naturally we must measure stability in relation to some other quantity; an absolute guarantee of stability cannot be achieved in constant time. We have chosen to measure the number of unstable edges relative to $\mysize{M}$, as this gives us the strongest results when compared with other quantities such as $\mysize{R}$, $\mysize{B}$, $\mysize{R \cup B}$, $\mysize{E}$, and $\mysize{R}\mysize{B}$, each of which is at least as large as $\mysize{M}$. Eriksson and Häggström \cite{eriksson08instability} provide arguments in favour of comparing unstable edges to $\mysize{E}$.

\subsection{Marriages vs.\ Roommates}\label{sec:bipartite}

An immediate question is whether our results apply to the stable roommates problem -- the non-bipartite version of the stable marriage problem. It is easy to see that a non-bipartite graph need not have an $\epsilon$-stable matching for $\epsilon<1$: consider a triangle with ``rotational-symmetric'' preferences -- there is one unstable edge to every matching edge.
\begin{center}
    \begin{picture}(0,0)%
\includegraphics{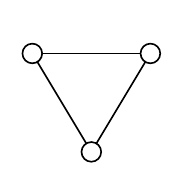}%
\end{picture}%
\setlength{\unitlength}{4144sp}%
\begingroup\makeatletter\ifx\SetFigFontNFSS\undefined%
\gdef\SetFigFontNFSS#1#2#3#4#5{%
  \reset@font\fontsize{#1}{#2pt}%
  \fontfamily{#3}\fontseries{#4}\fontshape{#5}%
  \selectfont}%
\fi\endgroup%
\begin{picture}(834,834)(-146,242)
\put(136,389){\makebox(0,0)[rb]{\smash{{\SetFigFontNFSS{8}{9.6}{\familydefault}{\mddefault}{\updefault}{\color[rgb]{0,0,0}$1$}%
}}}}
\put( 91,884){\makebox(0,0)[lb]{\smash{{\SetFigFontNFSS{8}{9.6}{\familydefault}{\mddefault}{\updefault}{\color[rgb]{0,0,0}$1$}%
}}}}
\put(451,884){\makebox(0,0)[rb]{\smash{{\SetFigFontNFSS{8}{9.6}{\familydefault}{\mddefault}{\updefault}{\color[rgb]{0,0,0}$2$}%
}}}}
\put(-44,659){\makebox(0,0)[rb]{\smash{{\SetFigFontNFSS{8}{9.6}{\familydefault}{\mddefault}{\updefault}{\color[rgb]{0,0,0}$2$}%
}}}}
\put(406,389){\makebox(0,0)[lb]{\smash{{\SetFigFontNFSS{8}{9.6}{\familydefault}{\mddefault}{\updefault}{\color[rgb]{0,0,0}$2$}%
}}}}
\put(586,659){\makebox(0,0)[lb]{\smash{{\SetFigFontNFSS{8}{9.6}{\familydefault}{\mddefault}{\updefault}{\color[rgb]{0,0,0}$1$}%
}}}}
\end{picture}%

\end{center}

\subsection{Bicoloured vs.\ Bipartite}\label{sec:2coloured}

A graph is $2$-colourable if and only if it is bipartite. Given a global view of a bipartite graph, it is trivial to $2$-colour the graph. In the context of distributed algorithms this relationship is more subtle, however. Linial \cite{linial92locality} shows that there is no constant-time algorithm for finding a maximal matching in a cycle with $2n$ with vertices. Czygrinow \myetal{} \cite{czygrinow08fast} show that finding a constant-factor approximation to maximum-cardinality matching with a deterministic constant-time algorithm is not possible either. However, these negative results heavily rely on the fact that the nodes in the bipartite graph do not know their colours in a 2-colouring. Theorem~\ref{thm:weighted-matching} shows that knowledge of the colour is not only necessary but also sufficient to find a constant-factor approximation of maximum-weight matching with a deterministic distributed constant-time algorithm.

The approximation factor of $2+\epsilon$ in Theorem~\ref{thm:weighted-matching} is almost the best possible for {\em any} algorithm (distributed or centralised) that has access only to the relative order of the weights, and not their numerical values. To see that no algorithm can have the approximation factor $2-\epsilon$ for any $\epsilon>0$, consider the following graph with the weight $\alpha > 1$.
\begin{center}
    \begin{picture}(0,0)%
\includegraphics{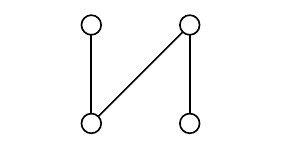}%
\end{picture}%
\setlength{\unitlength}{4144sp}%
\begingroup\makeatletter\ifx\SetFigFontNFSS\undefined%
\gdef\SetFigFontNFSS#1#2#3#4#5{%
  \reset@font\fontsize{#1}{#2pt}%
  \fontfamily{#3}\fontseries{#4}\fontshape{#5}%
  \selectfont}%
\fi\endgroup%
\begin{picture}(1284,654)(-416,287)
\put(-44,569){\makebox(0,0)[rb]{\smash{{\SetFigFontNFSS{8}{9.6}{\familydefault}{\mddefault}{\updefault}{\color[rgb]{0,0,0}$1$}%
}}}}
\put(496,569){\makebox(0,0)[lb]{\smash{{\SetFigFontNFSS{8}{9.6}{\familydefault}{\mddefault}{\updefault}{\color[rgb]{0,0,0}$1$}%
}}}}
\put(271,704){\makebox(0,0)[rb]{\smash{{\SetFigFontNFSS{8}{9.6}{\familydefault}{\mddefault}{\updefault}{\color[rgb]{0,0,0}$\alpha$}%
}}}}
\end{picture}%

\end{center}
If $\alpha \gg 1$, an algorithm must include the edge with weight $\alpha$ into the output (otherwise the algorithm has no approximation guarantee at all). But then the algorithm, having access only to the relative order of the weights, must include this edge also when $1 < \alpha < 2/{(2-\epsilon)}$, which contradicts with the approximation guarantee $2-\epsilon$.

\section{Related Work}\label{sec:related}

The stable marriage problem remains a subject of active research. The book by Gusfield and Irving \cite{gusfield89stable} is a comprehensive survey on the problem; see the MATCH-UP workshop \cite{halldorsson08match-up} for the latest developments.

\subsection{Almost Stable Matching}

Abraham \myetal{} \cite{abraham06almost-stable} study almost stable matchings in the stable roommates problem. The recent work by Biró \myetal{} \cite{biro08size} is particularly close to ours: they, too, consider the stable marriage problem with incomplete preference lists, and aim at finding a matching with few unstable edges. However, in terms of computational complexity, their work goes in the opposite direction.  Their task is to find a \emph{maximum} matching that minimises the number of unstable edges. It turns out that this makes the problem computationally much more \emph{difficult}: the problem is NP-hard, unlike the classical stable marriage problem. In contrast, we do not require that the matching is a maximum matching, which makes the problem computationally \emph{easier}: the problem admits a constant-time distributed algorithm, unlike the classical stable marriage problem. The algorithm works even when ties in the preferences lists are allowed; this should be contrasted with the fact that if ties are allowed, it is NP-hard to find a stable perfect matching \cite{halldorsson03approximability,irving08stable,iwama99stable}.

\subsection{Distributed Stable Matching}

Everyone witnesses evolution of matchings in real-world and virtual-world social networks; in its simplest form, the matchings attempt to attain stability by switching along unstable edges. In a seminal work, Knuth \cite{knuth76mariages} showed that such a switching may go in cycles, never resulting in a stable matching. Roth and Vande Vate \cite{roth90random} showed that switching {\em randomly} almost surely leads to stability.

Note that switching partners along unstable edges can be done in a distributed manner. The Gale--Shapley algorithm is also parallel by its nature: the proposals/rejects can be undertaken by all men/women simultaneously during synchronised {\em rounds} (albeit it can happen that only one man is free at a round \cite[Section~A.3]{gusfield89stable},\cite{tseng84parallel}). Lower bounds on the running time of the algorithm \cite[Section~1.5]{gusfield89stable} show that a linear number of rounds is required to attain stability. But can a {\em nearly} stable matching be obtained with fewer rounds?

Several works have addressed the last question with experiments. Quinn \cite{quinn85note} observes experimentally that a matching with only a fraction of unstable edges emerges long before the Gale--Shapley algorithm converges. Lu and Zheng \cite{lu03parallel} propose a parallel algorithm that outperforms the Gale--Shapley algorithm in practice. Theorem~\ref{thm:almost-stable} gives theoretical support to the findings in Quinn \cite{quinn85note}. Theorem~\ref{thm:almost-stable} also addresses the concern expressed in the conclusions of Lu and Zheng \cite{lu03parallel} where it is claimed that ``Most of existing parallel stable matching algorithms cannot guarantee a matching with a small number of unstable pairs within a given time interval.'' Theorem~\ref{thm:almost-stable} suggests that if the number of acceptable partners for each participant is bounded, the Gale--Shapley algorithm guarantees a small {\em relative} number of unstable edges.

From a theory perspective, apparently only a few papers address decentralised implementations of the Gale--Shapley algorithm and/or stability after early termination of a stable matching algorithm. In a recent paper \cite{eriksson08instability} it is claimed that ``little theory exists concerning instability.'' Khuller et al. \cite{khuller94on-line} give bounds on the performance of a simple online algorithm. Feder \myetal{} \cite{feder00sublinear} show that a stable matching can be found on a polynomial number of processors in sublinear time; their algorithm is not local. Other work on parallel stable matching include Tseng and Lee \cite{tseng84parallel}, Tseng \cite{tseng89average}, and Hull \cite{hull84parallel}. Eriksson and H\"{a}ggstr\"{o}m \cite{eriksson08instability} prove that a simple heuristic works well for {\em random} inputs. Our Theorem~\ref{thm:almost-stable} shows that the Gale--Shapley algorithm works well for an arbitrary input.

\subsection{Local Algorithms and Matchings}

As we mentioned in Section~\ref{sec:2coloured}, there is a range of negative results related to local algorithms (constant-time distributed algorithms) for maximal matching \cite{linial92locality} and approximate maximum matching \cite{czygrinow08fast,kuhn05price,kuhn06price,moscibroda06locality}. Even if each node is assigned a unique identifier and the network topology is an $n$-cycle, it is not possible to break the symmetry in the network and find a constant-factor approximation for maximum matching. Without any auxiliary information beyond unique node identifiers, positive results are known only in rare special cases, most notably for graphs where each node has an odd degree \cite{mayer95local,naor95what}.

Hence some auxiliary information is needed. For example, in problems related to stable marriage, it is natural to assume that every participant knows his/her gender. This assumption has not been exploited much in prior work on distributed, deterministic constant-time algorithms. We are only aware of Ha\'{n}\'{c}kowiak \myetal{} \cite{hanckowiak98distributed}, which shows that there is a constant-time algorithm for maximal matching in bicoloured bounded-degree graphs. Similarly to our algorithm, their algorithm does not need unique node identifiers; \emph{port numbering} \cite{angluin80local} is sufficient.

Other work on constant-time distributed algorithms for matching usually assumes either randomness \cite{hoepman06efficient,kuhn05price,kuhn06price,nguyen08constant-time,wattenhofer04distributed} or geometric information \cite{hassinen08analysing,wiese08local-matching}. We refer to the survey \cite{suomela08survey} for further information on local algorithms.

\subsection{Centralised Constant-Time Algorithms and Matchings}

Our centralised constant-time approximation algorithm in Theorem~\ref{thm:size-estimate} is based on the ideas of Parnas and Ron \cite{parnas07approximating} and Nguyen and Onak \cite{nguyen08constant-time}. Their work presents constant-time approximation algorithms for estimating the size of a maximal matching, maximum-cardinality matching, and maximum-weight matching. Our work complements this line of research by presenting an algorithm for estimating the size of a stable matching.

\section{Algorithm}\label{sec:algorithm}

We work with the following distributed variant of the Gale--Shapley algorithm. Each blue node $b \in B$ maintains the variable $p(b)$ which is the match of $b$ in the current matching or $\bot$ ($b$ is unmatched). Each red node $r \in R$ maintains the following variables: $C(r)$ is the list of candidates for matching; $c(r)$ is the candidate from whom $r$ is waiting for a response or $\bot$ ($r$ has no current candidate); and $p(r)$ is the match of $r$ in the current matching or $\bot$ ($r$ is unmatched).

The algorithm executes in {\em rounds}. Each round consists of two {\em turns}, a {\em blue turn} followed by a {\em red turn}. Each node $v\in R \cup B$ is active during turns of its colour.

\paragraph{Blue turn.}
Each blue node $b\in B$ completes the following read--compute--write cycle.
Initially, $p(b)=\bot$.

\begin{enumerate}
    \item Receive all incoming messages; let $P$ be the subset of neighbours that sent the message `propose'. If $P$ is empty, do nothing during this turn.
    \item If $p(b)\neq\bot$, then set $Q\gets P\cup\{p(b)\}$; otherwise set $Q\gets P$.
    \item Let $q$ be the node in $Q$ most preferred by $b$.
    \item If $q\neq p(b)$:
    \begin{enumerate}
        \item If $p(b) \ne \bot$ then send the message `break' to $p(b)$.
        \item Send the message `accept' to $q$.
        \item Set $p(b) \gets q$.
    \end{enumerate}
    \item For each $r \in P \setminus \{q\}$:
    \begin{enumerate}
        \item Send the message `reject' to $r$.
    \end{enumerate}
\end{enumerate}

\paragraph{Red turn.}
Each red node $r\in R$ completes the following read--compute--write cycle.
Initially, $C(r)$ consists of all adjacent blue nodes, in order of decreasing preference, most preferred first; $c(r) = \bot$; $p(r) = \bot$.

\begin{enumerate}
    \item If $c(r) \ne \bot$:
    \begin{enumerate}
        \item Receive a message $m$ from $c(r)$.
        \item If $m = \text{`accept'}$ then $p(r) \gets c(r)$.
        \item If $m = \text{`reject'}$ then remove $c(r)$ from $C(r)$.
        \item $c(r) \gets \bot$.
    \end{enumerate}
    \item If $p(r) \ne \bot$:
    \begin{enumerate}
        \item Receive a message $m$ from $p(r)$, if any.
        \item If $m = \text{`break'}$ then remove $p(r)$ from $C(r)$, and set $p(r) \gets \bot$.
    \end{enumerate}
    \item If $p(r) = \bot$ and $C(r)$ is not empty:
    \begin{enumerate}
        \item Let $c(r)$ be the first element of $C(r)$.
        \item Send the message `propose' to $c(r)$.
    \end{enumerate}
\end{enumerate}

\section{Analysis}\label{sec:analysis}

We adopt the convention of using a subscript $i=1,2,\ldots$ to denote the state of the algorithm at the end of round $i$. For example, $p_i(r)$ denotes the value of the variable $p(r)$ in the local state of $r$ at the end of round $i$.

At the end of round $i=1,2,\ldots$,
the local state variables $p_i(\cdot)$ define a matching
\[
\begin{split}
M_i&=\{\{r,p_i(r)\}:r\in R,\ p_i(r)\neq\bot\}\\
   &=\{\{b,p_i(b)\}:b\in B,\ p_i(b)\neq\bot\}\subseteq E.
\end{split}
\]

\subsection{Lost edges and convergence to stability}
We start from a restatement, in our terms, of the fact that as the Gale--Shapley algorithm progresses, women only improve their match and men only get worse.
\begin{lemma}
\label{lem:unstable}
An edge $\{r,b\}\in E\setminus M_i$ with $r\in R$ and $b\in B$ is unstable in $M_i$ only if $b$ is present in $C_i(r)$ and $r$ is unmatched at the end of round $i$.
\end{lemma}
\begin{proof}
If $b$ is not in $C_i(r)$, then either $b$ rejected $r$ or broke up with $r$ earlier, which means that $b$ prefers its current match to $r$.  If $b$ is in $C_i(r)$ and $r$ is matched, then $r$ has not proposed to $b$ yet, which means that $r$ prefers its current match to~$b$.
\end{proof}

During each round, each red node $r \in R$ removes at most one blue node $b\in B$ from $C(r)$. We say that such an edge $\{r,b\}$ is a \emph{lost edge} because it can no longer occur in the matching. We write $L_i \subseteq E$ for the set of all edges lost by the end of round $i$. Note that $L_1=\emptyset$ because only `propose' messages are sent during the first round.

Because $E$ is a finite set and $L_{i-1}\subseteq L_i\subseteq E$, there exists a $z\in\{2,3,\ldots\}$ such that $L_z=L_{z-1}$. Since received `reject' and `break' messages increase the number of lost edges, it follows that no such messages were received during round $z$. In particular, any unmatched $r\in R$ at the end of round $z$ must have $C_{z-1}(r)$ empty (and hence $C_z(r)$ empty), because otherwise it would have received a `reject' or `break' during round $z$. It follows by Lemma~\ref{lem:unstable} that $M_z$ is a stable matching. Let us denote this stable matching by $M_\infty$ in what follows.

\subsection{Weight and potential}

Associate with each edge $e\in E$ a positive integer weight $w(e)$ such that the weights respect the preferences of each node; that is, whenever a node $v$ prefers $x$ over $y$, it holds that $w(\{v,x\})\geq w(\{v,y\})$. This is always possible, e.g., by choosing $w \equiv 1$. In fact, to prove Theorem~\ref{thm:almost-stable}, taking $w \equiv 1$ is sufficient, while the proof (and actually the statement) of Theorem~\ref{thm:weighted-matching} deals with an arbitrary~$w$.

Associate with each blue node $b\in B$ the \emph{weight} $w_i(b) = w(\{b,p_i(b)\})$ if $b$ is matched; otherwise $w_i(b)=0$. The \emph{total blue weight}, $w_i(B) = \sum_{b \in B} w_i(b)$, is equal to the total weight of the matching $M_i$.

Associate with each red node $r\in R$ a \emph{potential} $f_i(r)$ as follows. If $r$ is matched or $C_i(r)$ is empty, set $f_i(r) = 0$. Otherwise, set $f_i(r) = w(\{r,b\})$ where $b$ is the first element of $C_i(r)$.

Intuitively, the potential $f_i(r)$ is an upper bound for the extra weight that $r$ could have if we ran the algorithm further than $i$ rounds. In a stable matching, $f_i(r) = 0$ for all $r \in R$, but we do not necessarily achieve this by running the algorithm for a constant number of rounds. However, we can derive an upper bound for the \emph{total potential} $f_i(R) = \sum_{r \in R} f_i(r)$ by observing that lost edges are heavier than those along which proposals (will) go.

\begin{lemma}
\label{lem:pot-ub}
For all $i=2,3,\ldots$ it holds that $f_i(R)\leq w(L_i)-w(L_{i-1})$.
\end{lemma}
\begin{proof}
Each $r \in R$ that has a positive potential at the end of round $i\geq 2$ has either received a `reject' or `break' during the round $i$; in both cases $r$ has removed a unique blue node $b$ from $C(r)$. The weight of the lost edge $\{r,b\}$ is at least the potential $f_i(r)$ because $C(r)$ is ordered by decreasing preference. The claim follows by taking the sum over all red nodes and edges lost during round $i$.
\end{proof}

\begin{lemma}
\label{lem:pot-dec}
For all $i=2,3,\ldots$ it holds that $f_i(R) \leq f_{i-1}(R)$.
\end{lemma}
\begin{proof}
Consider the red turn at round $i$. The potential of a red node $r\in R$ changes only if it receives a message. First, if $r$ receives a `reject' message during round $i$, the potential of $r$ at end of round $i$ may change but it does not increase, that is, $f_i(r) \le f_{i-1}(r)$. Second, if $r$ receives a `break' message from a blue node $b$, there is a unique red node $q$ that receives an `accept' message from $b$. In this case the potential of $q$ decreases from $f_{i-1}(q) = w(\{q,b\})$ to $f_i(q)=0$ while the potential of $r$ increases from $f_{i-1}(r) = 0$ to $f_i(r) \le w(\{r,b\}) \le w(\{q,b\}) = f_{i-1}(q)$. Third, if $r$ receives an isolated `accept' message (without an associated `break' message), $f_i(r)=0\leq f_{i-1}(r)$.  The claim follows by taking the sum over all `reject' messages, all `break'--`accept' matches, and all isolated `accept' messages during round~$i$.
\end{proof}

\begin{lemma}
\label{lem:lost-lb}
For all $i=2,3,\ldots$ it holds that $w(L_i) \geq (i-1) f_i(R)$.
\end{lemma}
\begin{proof}
By Lemma~\ref{lem:pot-ub} and Lemma~\ref{lem:pot-dec}, we have
\[
    w(L_i)=w(L_i)-w(L_1)=\sum_{j = 2}^i w(L_j)-w(L_{j-1}) \ge \sum_{j = 2}^i f_j(R) \ge (i-1) f_i(R).
\]
\end{proof}

\begin{lemma}
\label{lem:lost-ub}
For all $i=2,3,\ldots$ it holds that $w(L_i) \leq (\Delta-1) w_i(B)$
\end{lemma}
\begin{proof}
A blue node $b \in B$ can lose an incident edge $\{r,b\}$ during round $j = 2, 3, \ldots, i$ only if $b$ is matched at the end of round $j$ with a node that $b$ prefers to $r$. Put otherwise, if $\{r,b\} \in L_j\setminus L_{j-1}$, then $w(\{r,b\}) < w_j(b) \le w_i(b)$. The last inequality follows because each $b\in B$, once matched, only changes to a more preferred match. Furthermore, a blue node can lose at most $\Delta-1$ incident edges.
\end{proof}

\begin{lemma}
\label{lem:pot-ub2}
For all $\gamma>0$ and $i\geq 1+(\Delta-1)/\gamma$
it holds that $f_i(R) \le \gamma w_i(B)$.
\end{lemma}
\begin{proof}
By Lemma~\ref{lem:lost-lb} and Lemma~\ref{lem:lost-ub}.
\end{proof}

\section{Applications}\label{sec:applications}

\subsection{Proof of Theorem~\ref{thm:almost-stable}}

For the purpose of analysis, set $w(e) = 1$ for all $e \in E$. Then $\mysize{M_i}=w_i(B)$, so Lemma~\ref{lem:pot-ub2} and $i\geq 1+(\Delta-1)/\gamma$ imply $f_i(R)\leq\gamma\mysize{M_i}$. Denote by $u_i$ the number of unstable edges in $E\setminus M_i$. Because $f_i(R)$ counts the number of unmatched red nodes with a nonempty $C_i(r)$, it follows from  Lemma \ref{lem:unstable} that $u_i\leq \Delta f_i(R)$. Thus, $\gamma=\epsilon/\Delta$ gives $u_i\leq \epsilon\mysize{M_i}$ for $i\geq 1+\Delta(\Delta-1)/\epsilon$.

We can choose an integer $i < 2+\Delta^2/\epsilon$. We can compute $M_i$ in $i$ rounds or $T = 2i < 4 + 2\Delta^2/\epsilon$ synchronous communication steps. \qed

The proof for the case when ties in the preference lists are allowed, is analogous.

\subsection{Proof of Theorem~\ref{thm:weighted-matching}}
\label{sec:weighted-matching}

Let a positive integer weight $w(e)$ for each edge $e\in E$ be given as input. Assign the preferences of each node $v$ so that whenever $w(\{v,x\})>w(\{v,y\})$ for two edges $\{v,x\},\{v,y\}\in E$, $v$ prefers $x$ to $y$. Execute the algorithm in Section~\ref{sec:algorithm} for $i$ rounds to obtain the matching $M_i$. Let $M^{*} \subseteq E$ be a maximum-weight matching.

For each red node $r \in R$, let $g(r) \subseteq B$ consist of the matches of $r$ in $M_i$ and in $M^*$, if any. In particular, $\mysize{g(r)} \le 2$ for all $r \in R$, and $\mysize{g^{-1}(b)} \le 2$ for all $b \in B$.

Let $\{r,b\} \in M^{*}$ with $r\in R$ and $b\in B$. Exactly one of the following holds at the end of round $i$:
\begin{enumerate}
    \item The node $r$ has not received a response from $b$. If $r$ is matched with $b_2\neq b$, $r$ prefers $b_2$ over $b$ and $w(\{r,b\}) \leq w(\{r,b_2\}) = w_i(b_2)$. Otherwise $r$ is unmatched, $b$ is in $C_i(r)$, and $w(\{r,b\}) \leq f_i(r)$.
    \item The node $r$ has received a response from $b$. Then $b$ is matched with $r$ or with some other red node that $b$ prefers over $r$. Thus $w(\{r,b\}) \le w_i(b)$.
\end{enumerate}
For every $\{r,b\} \in M^{*}$ thus
\[
    w(\{r,b\}) \le f_i(r) + \sum_{c \in g(r)} w_i(c)
\]
and hence
\[
    \begin{split}
        w(M^{*}) &= \sum_{\{r,b\} \in M^{*}} w(\{r,b\}) \\
            &\leq \sum_{r \in R} f_i(r) + \sum_{r \in R} \sum_{c \in g(r)} w_i(c) \\
            &\le \sum_{R \in R} f_i(r) + \sum_{b \in B} 2 w_i(b) \\
            &= f_i(R) + 2 w_i(B).
    \end{split}
\]
From Lemma~\ref{lem:pot-ub2} we conclude that
\[
    w(M^{*}) \le (2 + \epsilon) w_i(B) = (2 + \epsilon) w(M_i)
\]
whenever $i\geq 1+(\Delta-1)/\epsilon$.

We can choose an integer $i < 2+\Delta/\epsilon$. We can compute $M_i$ in $i$ rounds or $T = 2i < 4 + 2\Delta/\epsilon$ synchronous communication steps. \qed

\subsection{Proof of Theorem~\ref{thm:size-estimate}}
\label{sec:size-estimate}
Let us first recall a convenient Chernoff-type upper bound for the tail probability of a binomial random variable $Z\sim\text{Bin}(k,p)$, where $k$ is a positive integer and $0\leq p\leq 1$. For $0\leq\beta\leq 1$, we have
\begin{equation}
\label{eq:chernoff}
\Pr\bigl(|Z-pk|\geq \beta pk\bigr)\leq 2\exp\biggl(-\frac{\beta^2pk}{3}\biggr).
\end{equation}
(See Janson \myetal{} \cite[Corollary~2.3]{janson00random} for a proof.)

Let $0<\delta\leq 1/2$, $0 < \epsilon \le 1$, and $\Delta \ge 3$ be given. For the purpose of analysis, set $w(e)=1$ for all $e\in E$. We estimate the size of the stable matching $M_\infty$ in $\myG$ using a randomised algorithm that queries the preference oracle for $\myG$. We assume that the number of nodes, $n$, is given as input to the algorithm.

Let us first relate $\mysize{M_j}$ to $\mysize{M_\infty}$. We have $\mysize{M_j}\leq \mysize{M_\infty}$ because every blue node, once matched, remains matched. Furthermore, $\mysize{M_\infty}\leq\mysize{M_j}+f_j(R)$ because every red node $r\in R$ with empty $C_j(r)$ will be unmatched in $M_\infty$. Let
\begin{equation}\label{eq:gamma-j}
    \gamma = \frac{\epsilon}{8 \Delta},
    \quad
    j \ge 1 + \frac{2\Delta-2}{\epsilon} = 1 + \frac{\Delta-1}{4\Delta\gamma}.
\end{equation}
Note that $\gamma < 1$. By $w_j(B)=\mysize{M_j}$ and Lemma~\ref{lem:pot-ub2} we have
\begin{equation}\label{eq:minf-mj-bound}
    \mysize{M_j}\leq\mysize{M_\infty}
    \leq (1 + 4\Delta\gamma)\mysize{M_j}
    \leq \mysize{M_j} + \frac{\epsilon}{2}\cdot\mysize{M_\infty}.
\end{equation}
It thus suffices to have an estimate for~$\mysize{M_j}$.

Because $\myG$ has no isolated nodes, $\mysize{R}\leq \Delta\mysize{B}$ and $\mysize{B}\leq\Delta\mysize{R}$. Because a stable matching is maximal and there are no isolated nodes,
\begin{align}
    \Delta\mysize{M_\infty} &\ge \mysize{R}, \label{eq:minf-r-bound}\\
    \Delta\mysize{M_\infty} &\ge \mysize{B}. \label{eq:minf-b-bound}
\end{align}
By \eqref{eq:minf-mj-bound} we have $\mysize{M_\infty} \le 2 \mysize{M_j}$ and thus
\begin{equation}\label{eq:mj-r-bound}
    \frac{\mysize{R}}{2\Delta}\leq \mysize{M_j}\leq \mysize{R}.
\end{equation}
Since $n=\mysize{R}+\mysize{B}$, we have
\begin{equation}\label{eq:r-n-bound}
    \frac{n}{\Delta+1}\leq \mysize{R}\leq n.
\end{equation}

Let
\begin{equation}\label{eq:N}
    N_0 = \frac{6 (\Delta+1)}{\gamma^2}\ln\frac{6}{\delta},
    \quad
    N\geq N_0.
\end{equation}
We use the following procedure to estimate $\mysize{M_j}$. First, select $N$ nodes uniformly at random and query the oracle for their colour. Denote by $X$ the number of red nodes among the $N$ nodes. Then, select $N$ red nodes uniformly at random; that is, select $2(\Delta+1)N$ nodes uniformly at random and select the first $N$ red nodes among those (if not enough red nodes occur, output `failure' and stop); denote by $Z$ the number of red nodes obtained (that is, failure occurs if and only if $Z<N$). For each red node, we determine whether it is matched in $M_j$; denote by $Y$ the number of selected red nodes that are matched in $M_j$. Output
\[
\hat{m}=n\cdot \frac{X}{N}\cdot \frac{Y}{N}
\]
and stop.

To analyse the procedure, let us first derive upper bounds for the probabilities of certain ``bad'' events.

Let us first derive an upper bound for the probability of failure. Let
\[
p=\frac{|R|}{n},\quad
k=2(\Delta+1)N,\quad\text{and}\quad
\beta=\frac{1}{2p(\Delta+1)}.
\]
By \eqref{eq:r-n-bound} we have $\beta\leq 1/2$. Furthermore, by \eqref{eq:r-n-bound} we have that $Z<N$ implies $|Z-pk|\geq N=\beta pk$. Thus, from \eqref{eq:chernoff}, \eqref{eq:r-n-bound}, and \eqref{eq:N} it follows that
\[
\begin{split}
\Pr(Z<N)
&\leq 2\exp\biggl(-\frac{\beta^2pk}{3}\biggr)\\
&=2\exp\biggl(-\frac{Nn}{6(\Delta+1)|R|}\biggr)\\
&\leq 2\exp\biggl(-\frac{N}{6(\Delta+1)}\biggr)
\leq \frac{\delta}{3}.
\end{split}
\]

Let us now derive an upper bound for the probability that the random variable $X/N$ significantly deviates from its expectation $|R|/n$.
Let
\[
p=\frac{|R|}{n},\quad k=N,\quad\text{and}\quad\beta=\gamma.
\]
Observe that $\beta\leq 1$ by \eqref{eq:gamma-j}. By \eqref{eq:chernoff}, \eqref{eq:r-n-bound}, and \eqref{eq:N}, we have
\[
\begin{split}
\Pr\biggl(\biggl|\frac{X}{N}-p\biggr|\geq\beta p\biggr)
&=\Pr\bigl(|X-pk|\geq\beta pk\bigr)\\
&\leq 2\exp\biggl(-\frac{\beta^2pk}{3}\biggr)\\
&= 2\exp\biggl(-\frac{\beta^2|R|N}{3n}\biggr)\\
&\leq 2\exp\biggl(-\frac{\beta^2N}{3(\Delta+1)}\biggr)
\leq \frac{\delta}{3}.
\end{split}
\]

Finally, let us now derive a similar bound for the probability that the random variable $Y/N$ significantly deviates from its expectation $|M_j|/|R|$. Let
\[
p=\frac{|M_j|}{|R|},\quad k=N,\quad\text{and}\quad\beta=\gamma.
\]
Observe that $\beta\leq 1$ by \eqref{eq:gamma-j}. By \eqref{eq:chernoff}, \eqref{eq:mj-r-bound}, and \eqref{eq:N}, we have
\[
\begin{split}
\Pr\biggl(\biggl|\frac{Y}{N}-p\biggr|\geq\beta p\biggr)
&=\Pr\bigl(|Y-pk|\geq\beta pk\bigr)\\
&\leq 2\exp\biggl(-\frac{\beta^2pk}{3}\biggr)\\
&=2\exp\biggl(-\frac{\beta^2|M_j|N}{3|R|}\biggr)\\
&\leq 2\exp\biggl(-\frac{\beta^2N}{6(\Delta+1)}\biggr)
\leq \frac{\delta}{3}.
\end{split}
\]

Next we show that, with probability at least $1-\delta$, the estimate $\hat m$ approximates $|M_\infty|$ with the claimed accuracy. Observe first that $|ac-bd|\leq a|c-d|+d|a-b|$ for $a,b,c,d\geq 0$. Thus, with probability at least $1-\delta$, both no failure occurs and
\begin{equation}
\label{eq:mj-gamma-bound}
\begin{split}
\bigl|\hat m-|M_j|\bigr|
&=n\biggl|\frac{X}{N}\cdot\frac{Y}{N}-\frac{|R|}{n}\cdot\frac{|M_j|}{|R|}\biggr|\\
&\leq n\biggl(\frac{X}{N}\cdot\biggl|\frac{Y}{N}-\frac{|M_j|}{|R|}\biggr|+
      \frac{|M_j|}{|R|}\cdot\biggl|\frac{X}{N}-\frac{|R|}{n}\biggr|\biggr)\\
&\leq n\biggl(\frac{X}{N}\cdot\gamma\cdot\frac{|M_j|}{|R|}+
      \frac{|M_j|}{|R|}\cdot\gamma\cdot\frac{|R|}{n}\biggr)\\
&\leq 2\gamma n
= 2 \gamma\bigl(\mysize{R} + \mysize{B}\bigr)
\leq 4 \Delta\gamma\mysize{M_\infty}
= \frac{\epsilon}{2}\cdot \mysize{M_\infty}
\end{split}
\end{equation}
where the last inequality follows by \eqref{eq:minf-r-bound} and \eqref{eq:minf-b-bound}. Thus, with probability at least $1-\delta$, we have
\[
\begin{split}
    \bigl|\hat m-|M_\infty|\bigr|
    &\leq \bigl|\hat m-|M_j|\bigr|+\bigl||M_j|-|M_\infty|\bigr|\\
    &\leq \frac{\epsilon}{2}\cdot\mysize{M_\infty} + \frac{\epsilon}{2}\cdot\mysize{M_\infty} = \epsilon \mysize{M_\infty}
\end{split}
\]
where the last inequality follows by \eqref{eq:minf-mj-bound} and \eqref{eq:mj-gamma-bound}.

It remains to derive an upper bound for the number of oracle queries we make to compute the estimate. First, we make $N$ queries to determine $X$. Second, we make $2(\Delta+1)N$ queries to obtain $N$ red nodes (or declare failure). Third, for each of the $N$ red nodes, we execute the algorithm in Section~\ref{sec:algorithm} for $j$ rounds to determine whether the node is matched in $M_j$. Each round of the algorithm propagates information (messages) for at most two hops in $\myG$, which implies that we can decide whether any given $r\in R$ is matched in $M_j$ if we know the preferences of all nodes within distance $2j$ from $r$ in $\myG$. There are at most
\[
    1+\Delta\sum_{i=0}^{2j-1}(\Delta-1)^i = 1 + \frac{\Delta}{\Delta-2} \left( (\Delta-1)^{2j} - 1 \right) < 3 (\Delta-1)^{2j}
\]
such nodes. Because $\Delta \ge 3$ and $\delta \le 1/2$, we have
\[
\begin{split}
    N_0 &= 384\Delta^2 (\Delta+1) \epsilon^{-2} \bigl(\ln 6 + \ln \delta^{-1}\bigr) \\
        &\le 384\cdot  \left(\frac{3}{2}\right)^{\!2} \cdot\frac{4}{2}\cdot (\Delta-1)^3 \epsilon^{-2} \left(\frac{\ln 6}{\ln 2} + 1\right) \ln \delta^{-1} \\
        &< 6195 (\Delta-1)^3 \epsilon^{-2} \ln \delta^{-1}.
\end{split}
\]
Therefore we can select integers $j$ and $N$ such that $j \le 2 \Delta \epsilon^{-1}$ and $N \le 6250 {(\Delta-1)^3} {\epsilon^{-2}} \ln \delta^{-1}$. In total we make at most
\[
\begin{split}
    N + 2(\Delta+1)N + 3(\Delta-1)^{2j} N &\le 4 (\Delta-1)^{2j} N \\
    &< 25000 \epsilon^{-2} (\Delta-1)^{3 + 4\Delta/\epsilon}\ln\delta^{-1}
\end{split}
\]
queries. This completes the proof. \qed

\section*{Acknowledgements}
This work was supported in part by the Academy of Finland, Grants 116547, 117499, and 118653 (ALGODAN), by Helsinki Graduate School in Computer Science and Engineering (Hecse), and by the Foundation of Nokia Corporation.

We acknowledge email exchange with the members of the Algorithms and Complexity Research at Glasgow group, and thank the organisers of ALGO 2008 for the creative atmosphere at the conference.

\end{document}